\documentclass[runningheads,envcountsame,envcountsect]{llncs}

\usepackage{graphicx}
\usepackage{algorithmic}
\usepackage{algorithm}
\usepackage{amsmath}
\usepackage{amssymb}
\usepackage{amsfonts}
\usepackage{amsxtra}
\usepackage{enumitem}
\usepackage{newtxtext}
\usepackage{newtxmath}
\usepackage{mathtools}
\usepackage{etoolbox}
\usepackage{booktabs}
\usepackage{adjustbox}
\usepackage{siunitx}

\usepackage[round]{natbib}

\spnewtheorem{fact}[theorem]{Fact}{\bfseries}{\itshape}
\numberwithin{theorem}{section}
\numberwithin{equation}{section}

\newcommand{\pcite}[1]{\citeauthor{#1}'s (\citeyear{#1})} 

\makeatletter
\let\c@algorithm\c@theorem

\makeatother

\usepackage[svgnames]{xcolor}
\definecolor{darkgreen}{rgb}{0,.35,0}
\definecolor{darkblue}{rgb}{0,0,.1}
\definecolor{darkred}{rgb}{.6,0,0}

\makeatletter
\newcommand{\papertitle}{\@title}  
\newcommand{\paperauthor}{\@author}
\makeatother

\makeatletter
\newcommand{\ReDeclareMathOperator}{%
  \@ifstar{\@declmathop\@empty}{\@declmathop o}}
\long\def\@declmathop#1#2#3{%
  \DeclareRobustCommand{#2}{\qopname\newmcodes@#1{#3}}}
\makeatother

\DeclareTextFontCommand{\boldemph}{\bfseries\em}

\newcommand{\softO}{{O\mskip1mu\tilde{\,}\mskip1mu}}

\newcommand{\ZZ}{{\mathbb{Z}}}
\newcommand{\FF}{{\mathbb{F}}}
\newcommand{\F}{{\mathsf{F}}}
\newcommand{\R}{{\mathsf{R}}}

\newcommand{\Rhat}{%
  \mathchoice
    {\smash{\tilde{\mathsf{R}}}\vphantom{\mathsf{R}}}
    {\smash{\tilde{\mathsf{R}}}\vphantom{\mathsf{R}}}
    {\smash{\tilde{\scriptstyle\mathsf{R}}}\vphantom{\scriptstyle\mathsf{R}}}
    {\smash{\tilde{\scriptscriptstyle\mathsf{R}}}\vphantom{\scriptscriptstyle\mathsf{R}}}%
}

\newcommand{\rank}{\operatorname{rank}}
\newcommand{\SNF}{\operatorname{SNF}}
\newcommand{\diag}{\operatorname{diag}}
\newcommand{\GL}{\operatorname{GL}}
\newcommand{\GR}{\operatorname{GR}}
\newcommand{\im}{\operatorname{im}}

\newcommand{\ceil}[1]{\left\lceil #1\right\rceil}

\newcommand{\nxn}{{n\times n}}
\newcommand{\trans}{\intercal}

\usepackage{hyperref}

\hypersetup{
  colorlinks=true,
  linkcolor=blue!50!black,
  citecolor=blue!50!black,
  urlcolor=blue!60!black,
  pdftitle={Computing Smith Forms Modulo p^2 of Sparse Matrices Faster Than Matrix Multiplication},
  pdfauthor={Mark Giesbrecht}
}

\newcommand{\tightdisplay}[1]{%
  \begingroup
  \setlength{\abovedisplayskip}{4pt}%
  \setlength{\belowdisplayskip}{4pt}%
  \setlength{\abovedisplayshortskip}{2pt}%
  \setlength{\belowdisplayshortskip}{2pt}%
  \[
  #1
  \]
  \endgroup\ignorespaces
}

\begin{document}

\title{Computing Smith Forms Modulo $p^2$ of Sparse Matrices Faster
  Than Matrix Multiplication}

\titlerunning{Smith Forms Modulo \texorpdfstring{$p^2$}{p^2}}

\author{Mark Giesbrecht}
\authorrunning{M. Giesbrecht}
\institute{Cheriton School of Computer Science, Faculty of Mathematics, University of Waterloo, Canada\\
  \email{mwg@uwaterloo.ca}}

\maketitle

\begin{abstract}
  Let $p$ be a prime and $\R=\ZZ/p^2\ZZ$ the ring of integers modulo
  $p^2$.  Any $A\in\R^\nxn$ is unimodularly equivalent to its Smith
  form
  \[
    S=\diag\bigl(\underbrace{1,\ldots,1}_{r_0},
    \underbrace{p,\ldots,p}_{r_1}, \underbrace{0,\ldots,0}_{r_2}\bigr)
    \in\R^\nxn,
  \]
  i.e., there exist $U,V\in\R^\nxn$ such that $UAV=S$, with
  $\det U,\det V\in\R^*$ (where $\R^*$ is the set of units in $\R$,
  elements not equivalent to $0\bmod p$).  Our goal in this paper is
  to determine $r_0,r_1,r_2$ quickly when $A$ is sparse or structured.
  By ``sparse'' we mean $A$ is given by a black box such that for any
  $v\in\R^{n\times 1}$ we can compute $v\mapsto Av$ with $\softO(n)$
  operations in $\R$, which captures having few nonzero elements or a
  multiplicative structure (e.g., Hankel or Toeplitz matrices).  We
  present a randomized algorithm which requires an expected number of
  \tightdisplay{
    \softO\left(n^{3-1/(\omega-1)}\right)
    \hspace*{10pt}\mbox{operations in $\R$}
  }
  to compute the Smith form, where $\omega$ is the exponent of dense
  matrix multiplication.  Using standard cubic matrix multiplication
  ($\omega=3$) our algorithm thus requires $\softO(n^{2.5})$
  operations in $\R$, while using the current asymptotically fastest
  matrix multiplication, with $\omega<2.371339$, our algorithm
  requires $\softO(n^{2.270786})$ operations in $\R$.  Our algorithm
  is probabilistic of the Monte Carlo type, meaning it fails on any
  invocation with controllably small probability.  We employ iterative
  block-Wiedemann-style matrix techniques and structured
  preconditioners.  To our knowledge, this is the first algorithm to
  compute the modular Smith Normal Form modulo $p^2$ requiring fewer
  than $\softO(n^\omega)$ operations in $\R$, i.e., faster than any
  dense algorithm.

\end{abstract}

\footnotetext{To appear,
  \href{https://www.casc-conference.org/index.html}{Computer Algebra
    in Scientific Computation (CASC)} conference, August
  31--September 4, 2026, Bath, UK}

\section{Introduction}

The Smith normal form ($\SNF$), introduced by \cite{Smi61},
diagonalizes matrices over $\ZZ$ and has applications in
diophantine analysis~\citep{ChoCol82}, integer
programming~\citep{Hu69}, combinatorics~\citep{Sta16}, determining the
structure of abelian groups \citep{New72} and class
groups~\citep{Haf89}, in control system theory~\citep{Kai80}, and
especially computing simplicial homology~\citep{DumHec03}.  To deal
with large instances of these problems we must exploit sparsity in the
matrix (see, e.g., \cite{GiuGhr16} in an application to neural
data). This remains challenging.

The problem of computing the Smith form of a sparse matrix over a
principal ideal ring presents several challenges.  The mathematics
goes back to \cite{Kap49}. One approach is to simply compute the SNF
over $\ZZ$ and then reduce the result modulo the power of the prime
ideal. The algorithm of \cite{EbeGie07} for SNF of a sparse matrix
over $\ZZ$ could be used, which reduces the cost using standard matrix
arithmetic over $\ZZ$ to $\softO(n^{2.66})$ bit operations, but does
not seem directly amenable to asymptotic speedup.  The best known
algorithm for any matrix over any $\ZZ/(d)$, for any $d\in\ZZ$, is by
\cite{Sto96} and requires $\softO(n^\omega)$ operations mod $d$, and
is not sensitive to sparsity.  The ``valence'' algorithm of
\cite{DumHec03} uses a hybrid of a ``global'' method over $\ZZ$ and a
(dense) $p$-adic elimination method for selected primes $p$.  Our
method here may be useful as a subroutine in this approach.

We will address a simple but important case here to attempt to break
through this barrier: integer matrices modulo $p^2$ for a prime $p$.
For matrices over $\R=\ZZ/p^2\ZZ$, there are only three possible Smith
factors, and the SNF has shape
\tightdisplay{
  S=\diag\bigl(\underbrace{1,\ldots,1}_{r_0},
  \underbrace{p,\ldots,p}_{r_1}, \underbrace{0,\ldots,0}_{r_2}\bigr)
  = \diag\bigl(1^{r_0},p^{r_1},0^{r_2}\bigr)
  \in\R^\nxn.
}
The first multiplicity is immediate:
\tightdisplay{
  r_0=\rank(A\bmod p),
}
which can be identified by reducing the matrix to the prime field
$\FF_p$ and using well-established algorithms such as \cite{KalSau91},
which require $\softO(n^2)$ operations in $\R$.  The difficult part is
to compute the number of invariant factors equal to $p$ without
forming dense nullspace bases, dense cokernel bases, or a dense Schur
complement, all of which seem expensive to compute, even for
sparse matrices.

The goal of this paper is, given a black box matrix over $\R$ as
above, to find the SNF multiplicities (i.e., $r_0,r_1,r_2$) with
fewer than $\softO(n^\omega)$ operations in $\R$.  We assume that two
$n\times n$ matrices can be multiplied with $O(n^\omega)$ operations,
where $\omega=3$ using classical methods and $\omega<2.371339$ using
the fastest current methods \cite{AlmDua25}.

The key new idea is to split off the unit part modulo $p$.  After fast
rank-profile preconditioning, we write an equivalent matrix as
\tightdisplay{
  A'=\begin{bmatrix}
    B&C\\
    D&E
  \end{bmatrix},
  \qquad B\bmod p~~\text{nonsingular},\qquad \dim B=r_0.
}
The Schur complement
\[
  J=E-DB^{-1}C \bmod p^2
\]
is divisible by $p$.  Thus $J=pT$ for a matrix $T$ over the residue
field, and
\tightdisplay{
  r_1=\rank(T),\qquad r_2=n-r_0-r_1.
}
The matrix $T$ is never formed.  To evaluate $TX$, lift $X$ to
$\widetilde X$, solve $BY=C\widetilde X\bmod p^2$, and return
\tightdisplay{
  \frac{E\widetilde X-DY}{p}\bmod p,
}
or equivalently the lower block of
\tightdisplay{
  A' \begin{bmatrix}-Y\\
    \widetilde X
  \end{bmatrix}
}
divided by $p$.  We show how to perform the solve by Hensel lifting,
with residuals computed in the length-two ring before division.

Another important tool we use is efficient sparse blocking.  The
expensive solves with $B\bmod p$ are carried out by the sparse-block
solve machinery of \cite{EbeGie06,EbeGie07}, which provides efficient
block projections.  To avoid relying on a sparse-projection version of
arbitrary-rank block Wiedemann, we compute rank by \pcite{Wie86}
binary-search reduction to nonsingularity tests, using the
PRECONDSXS/rank-profile conditioners of \cite{CheEbe02}.  Each
nonsingularity test is a verified sparse-block solve attempt.

In the nontrivial case $0<r_0<n$, the sparse-block rank computation
will use optimized block size
\[
  t=\min\{r_0,n-r_0,\lceil n^{1/(\omega-1)}\rceil\}.
\]
The resulting Monte Carlo complexity is
\tightdisplay{
  \softO\left(n^{3-1/(\omega-1)}\right)
  \hspace*{10pt}\mbox{operations in $\R$}.
}
Since $3-1/(\omega-1)<\omega$ for every $2<\omega\leq 3$, this is
asymptotically below dense $n^\omega$ arithmetic.  The rank
computations are Monte Carlo; the sparse-block linear solves inside
the algorithm are verified by substitution.

The problem is closely related to earlier algorithms for Smith forms
over local rings investigated in \citep{ElsGie12}.  The contribution
here is the $p^2$ specialization in which the matrix $T=J/p\bmod p$ can
be accessed by a carry-safe block oracle, and sparse-block solves
amortize the cost of the unit-block inversions enough to obtain a
uniform sub-$n^\omega$ bound for arbitrary sparse input matrices.

The remainder of the paper gives the local algebra, the finite-field
sparse-block rank/solve procedures, the carry-safe oracle, the correctness
proof, and the complexity analysis.  

\section{Notation, model, and output}

The base ring for our computations is $\R=\ZZ/p^2\ZZ$.  If the residue
field $\R/p\R$ is too small for the randomized field preconditioners,
we extend scalars to the unramified Galois ring $\Rhat=\GR(p^2,d)$,
whose residue field is $\F=\Rhat/p\Rhat\cong \FF_{p^d}$.  When no
extension is needed, take $d=1$ and $\Rhat=\R$, $\F=\FF_p$.  We will
always count operations over $\R$ (or $\F$, as less expensive than
$\R$).  Since any extension used has degree
$d=O(\log n + \log(1/\epsilon))$ we will assume these are captured by
the polylogarithmic constants in the $\softO(\cdot)$.

The maximal ideal of $\Rhat$ is $p\Rhat$ and $p^2=0$.  Multiplication
by $p$ induces an isomorphism of $\F$-vector spaces
$ \Rhat/p\Rhat \cong p\Rhat$, $\bar x\mapsto px$.  Thus if
$z\in p\Rhat$, the notation $z/p\bmod p$ is unambiguous.

For our black box model we assume only that, for any vector
$v\in\R^{n\times 1}$, we can evaluate $v\mapsto Av$ with $\softO(n)$
operations in $\R$ (and similarly over $\Rhat$).  For a block
$V\in\R^{n\times s}$, applying this product to each column gives
$V\mapsto AV$ in $\softO(ns)$ operations.

If the input black box is supplied only over the base ring $\R$ and
the algorithm works over $\Rhat$, fix an $\R$-basis of the free
unramified extension $\Rhat/\R$.  Since the entries of $A$ lie in
$\R$, a product with a vector in $\Rhat^n$ decomposes into $d$
products with vectors in $\R^n$, plus $\softO(nd)$ scalar operations
for basis conversion and recombination.  In the contexts below
$d=O(\log n+\log(1/\epsilon))$ for fixed or word-sized $p$, and this
overhead is absorbed into the $\softO(\cdot)$ bounds.

Over $\R$ or $\Rhat$, every Smith form modulo $p^2$ has the shape
$\diag(1^{r_0},p^{r_1},0^{r_2})$.  Our algorithm computes the
multiplicities $r_0,r_1,r_2$.  \emph{Note that we do not compute the
  unimodular transformation matrices.}

The following shows that moving to an enlarged ring does not change
the Smith form.
\begin{lemma}
  \label{lem:extend}
  Let $A\in \R^{n\times n}$ and extend scalars to $\Rhat=\GR(p^2,d)$.
  If
  \[
    \SNF_\R(A)=\diag(1^{r_0},p^{r_1},0^{r_2}),
  \]
  then the Smith form of $A$ over $\Rhat$ has the same multiplicities.
\end{lemma}
\begin{proof}
  Write a Smith decomposition over $\R$ as
  \tightdisplay{
    A=U\diag(1^{r_0},p^{r_1},0^{r_2})V,
    \qquad U,V\in\GL_n(\R).
  }
  After scalar extension, the same equality holds over $\Rhat$, and the
  images of $U$ and $V$ remain invertible because their determinants
  remain units.  Thus the displayed diagonal is a Smith diagonal for $A$
  over $\Rhat$.
  
  The ring $\Rhat$ is again a length-two principal local ring with
  maximal ideal generated by the same element $p$.  Hence its Smith
  diagonal entries are uniquely determined up to multiplication by
  units and are represented by $1$, $p$, and $0$.  Therefore the
  multiplicities of these three types are unchanged.  Notice that
  ranks in the extended algorithm are measured over
  $\F=\FF_{p^d}$; they are not multiplied by $d$ when viewed as
  $\F$-dimensions. \qed
\end{proof}

To avoid technicalities, the reader can generally assume that our ring
$\R$ has a sufficiently large residue field that the algorithms work
correctly, but explicit remediation for this is discussed below.

\section{The local algebra modulo $p^2$}

Let $\bar A=A\bmod p\in \F^{n\times n}$.  If
\tightdisplay{
  \SNF(A)=\diag(1^{r_0},p^{r_1},0^{r_2}),
}
then
\tightdisplay{
  \bar A\sim \diag(I_{r_0},0,0)
}
over $\F$.  Therefore
\[
  r_0=\rank_\F(\bar A).
\]
This part is immediate and is computed by a finite-field sparse rank
algorithm with $\softO(n^2)$ operations in $\R/(p)$.

Let
\[
  a=\rank_\F(\bar A),\qquad k=n-a.
\]
Suppose, after invertible row and column operations over $\Rhat$, we
have
\tightdisplay{
  A'=
  \begin{bmatrix}
    B&C\\
    D&E
  \end{bmatrix},
  \qquad B\in \Rhat^{a\times a},
}
with $\bar B=B\bmod p$ nonsingular.

\begin{lemma}[Schur complement reduction]
  \label{lem:schur}
  The matrix $B$ is invertible over $\Rhat$, and $A'$ is equivalent over
  $\Rhat$ to
  \tightdisplay{
    \begin{bmatrix}
      B&0\\
      0&J
    \end{bmatrix},
    \qquad J=E-DB^{-1}C\bmod p^2.
  }
  Moreover $J\equiv 0\bmod p$.
\end{lemma}

\begin{proof}
  Since $\bar B$ is nonsingular over the residue field, $\det(B)$ is a
  unit in the local ring $\Rhat$, so $B$ is invertible.  Multiplication
  by the unimodular block matrices
  \[
    \begin{bmatrix}
      I&0\\
      -DB^{-1}&I
    \end{bmatrix}
    \quad\text{and}\quad
    \begin{bmatrix}
      I&-B^{-1}C\\
      0&I
    \end{bmatrix}
  \]
  gives
  \[
    \begin{bmatrix}
      I&0\\
      -DB^{-1}&I
    \end{bmatrix}
    \begin{bmatrix}
      B&C\\
      D&E
    \end{bmatrix}
    \begin{bmatrix}
      I&-B^{-1}C\\
      0&I
    \end{bmatrix}
    =
    \begin{bmatrix}
      B&0\\
      0&E-DB^{-1}C
    \end{bmatrix}.
  \]
  Modulo $p$, the same block elimination is valid over $\F$.  Since
  $\rank(\bar A')=a$ and $\bar B$ has rank $a$, the lower-right Schur
  complement modulo $p$ must have rank zero.  Hence
  $J\equiv 0\bmod p$. \qed
\end{proof}

Thus, there is a unique matrix $T\in \F^{k\times k}$ such that $J=pT$.

\begin{lemma}[The Schur complement divided by $p$ yields $r_1$]
  \label{lem:divided-schur-rank}
  With the notation above,
  \tightdisplay{
    r_1=\rank_\F(T),\qquad r_2=k-\rank_\F(T).
  }
\end{lemma}

\begin{proof}
  The block $B$ is invertible over $\Rhat$, so it contributes exactly
  $a$ unit invariant factors.  It remains to compute the Smith form of
  $J=pT$.

  Let $b=\rank_\F(T)$.  Choose $U_0,V_0\in\GL_k(\F)$ such that
  \tightdisplay{
    U_0 T V_0=\diag(I_b,0).
  }
  Lift $U_0,V_0$ to matrices $\widetilde U,\widetilde V\in\GL_k(\Rhat)$,
  and choose any lift $\widetilde T\in\Rhat^{k\times k}$ of $T$.  Since
  multiplication by $p$ kills all terms divisible by $p$,
  \tightdisplay{
    \widetilde U\,(p\widetilde T)\,\widetilde V
    =p\,\widetilde{(U_0TV_0)} =\diag(pI_b,0)\pmod {p^2}.
  }
  Thus $pT$ has $b$ invariant factors equal to $p$ and $k-b$ invariant
  factors equal to $0$. \qed
\end{proof}

\section{Adapting the fast block matrix tools}
\label{sec:finite-field-engine}

This section adapts and summarizes the block matrix tools we need for
our new algorithm.  We make no claim of novelty here, but provide a
summary of methods adapted to our particular context.  Our finite
field routines are based on the algorithms of
\cite{EbeGie06,EbeGie07}, which we will sometimes refer to as
\emph{EGGSV}.  Rank is reduced to the binary-search nonsingularity
tests of \cite{Wie86}, using the preconditioners of \cite{CheEbe02},
and each nonsingularity test is implemented by an EGGSV sparse-block
solve attempt with a random right-hand side.

The only black-box operation in the Smith algorithm is a matrix-vector
product with the original matrix $A$, assumed to cost $\softO(n)$
operations in $\R$.  Whenever we need to apply $A$ to an $n\times s$
block, we apply the black box to the columns, at cost $\softO(ns)$.
The finite-field statements below are written for an auxiliary matrix
$M\in \F^{m\times m}$ for which one vector product $v\mapsto Mv$ costs
$\softO(L)$ operations.  Thus an $m\times s$ block product by $M$
costs $\softO(Ls)$ by columnwise application.  In the Smith algorithm
this will be used with $L=n$ for leading submatrices of the
preconditioned sparse matrix.  The divided Schur complement $T$ is not
treated as a new input black box.  Section~\ref{sec:T-oracle}
constructs its block oracle from the original black box for $A$ and
from the verified sparse-block solves with~$B_0$.

Throughout this section, $1\le s\le m$ is a block size.  If
$s\nmid m$, pad by fewer than $s$ dimensions.  For nonsingular solves
use a direct-sum identity block, and for rank tests use a direct-sum
zero block, or an identity block and subtract its known rank.  This
changes all costs by at most constant factors.  Thus, we write $m=qs$.
The sparse-block projection used by EGGSV is
\[
  u=\begin{bmatrix}
    I_s\\
    \vdots\\
    I_s
  \end{bmatrix}\in \F^{m\times s},
\]
so multiplication by $u$ or $u^\trans$ on an $m\times s$ block costs
$O(ms)$ operations, not $O(ms^2)$ operations in $\F$.

The sparse-block projection and inverse factorization method of
\cite{EbeGie06,EbeGie07} is at the foundation of our algorithm and
reduced complexity.  We summarize its adapted use, features and
complexity as follows.

\begin{theorem}[Efficient block projections and inverse factorization]
  \label{thm:eggsv-projections}
  Let $N\in \F^{m\times m}$ be nonsingular and $m=qs$.  After fast
  rank-profile conditioning over a sufficiently large field, the
  sparse projection $u$ above has the following properties with
  controllably high probability.
  \begin{enumerate}[label=(\roman*),leftmargin=2.5em]
  \item The right block Krylov matrix
    \[
      K^{(r)}=K_q(\widehat N,u)=[u,\widehat N u,\ldots,\widehat
      N^{q-1}u]
    \]
    and the block-row Krylov matrix
    \[
      K^{(\ell)}=
      \begin{bmatrix}u^\trans\\
        u^\trans\widehat N\\
        \vdots\\
        u^\trans\widehat N^{q-1}
      \end{bmatrix}
    \]
    are nonsingular.  Here $\widehat N$ is obtained from $N$ by fast
    invertible conditioners and a block diagonal random scaling.
  \item The block-Hankel matrix
    \[
      H=K^{(\ell)}\widehat N K^{(r)}
    \]
    is nonsingular and has $s\times s$ Hankel blocks.
  \item There is a factorization
    \[
      \widehat N^{-1}=K^{(r)}H^{-1}K^{(\ell)}.
    \]
    Undoing the fast conditioners gives the corresponding
    factorization of $N^{-1}$.
  \end{enumerate}
  Moreover, the conditioners apply to blocks in near-linear time, and
  the failure probability can be made at most $\epsilon$ by working
  over an extension field of degree $O(\log m+\log(1/\epsilon))$
  and/or by repetition.
\end{theorem}

\begin{proof}
  This is the sparse-block projection and inverse-factorization
  approach of \cite{EbeGie06,EbeGie07}, stated in our context and
  notation.  The sparse projection $u$ is the projection in
  \cite[(2.1)]{EbeGie07}.  The Krylov nonsingularity statements under
  leading-minor hypotheses are \cite[Theorem~2.1 and
  Corollary~2.2]{EbeGie07}; the passage to arbitrary nonsingular
  matrices over sufficiently large fields by fast rank-profile
  conditioning is \cite[Corollary~2.3 and Section~4]{EbeGie07}.  The
  block-Hankel inverse factorization is \cite[Theorem~3.1]{EbeGie07},
  and the Las Vegas inversion and inverse-times-matrix consequences
  are \cite[Theorem~4.2 and Corollary~4.4]{EbeGie07}.  \qed
\end{proof}

The following is the parameterized form of the EGGSV
inverse-times-matrix algorithm, specialized to one $m\times s$
right-hand-side block and to the sparse projection
$u=[I_s,\ldots,I_s]^\trans$.

\begin{theorem}
  \label{thm:param-solve}
  Let $M\in \F^{m\times m}$ be nonsingular, and suppose one vector
  product $v\mapsto Mv$ costs $\softO(L)$ operations.  For
  $C\in \F^{m\times s}$, there is a Las Vegas algorithm that computes
  $M^{-1}C$ with expected cost
  \[
    \softO\left(mL+m^2+m s^{\omega-1}\right).
  \]
\end{theorem}

\begin{proof}
  This is \cite[Corollary~4.4]{EbeGie07}, parameterized by the block
  size $s$ and applied to one $m\times s$ right-hand-side block.  Let
  $m=qs$.  The EGGSV construction uses the sparse projection
  $u=[I_s,\ldots,I_s]^\trans$, forms the block-Hankel Krylov data
  $u^\trans M^iu$, applies $K^{(\ell)}$ to the right-hand-side block
  $C$, solves the resulting structured block-Hankel system, and
  evaluates $K^{(r)}$ on the result.

  The cost is assessed as follows.  The $O(q)$ Krylov block products
  needed for the Hankel data, for $K^{(\ell)}C$, and for the final
  $K^{(r)}$ evaluation each cost $\softO(Ls)$ by $s$ columnwise vector
  products.  These contributions are $\softO(mL)$.  The sparse
  projections by $u^\trans$ over the generated sequences contribute
  $\softO(m^2)$ field operations.  The structured $q\times q$
  block-Hankel solve has $s\times s$ blocks and costs
  $\softO(qs^\omega)=\softO(ms^{\omega-1})$.  Fast conditioning and
  unconditioning are absorbed in these terms.

  The returned candidate is verified by checking $MX=C$.  Since $M$ is
  nonsingular, bad internal choices are detected either by a singular
  Krylov/Hankel construction or by this final verification, and
  independent resampling gives the stated expected cost. \qed
\end{proof}

The full inverse algorithm of \cite{EbeGie07} is stated with
access to both $M$ and $M^\trans$.  In the present paper we only use
the inverse-times-a-block variant.  For this task, the product
$K^{(\ell)}C$ is obtained from the Krylov sequence
\tightdisplay{
  C,\; MC,\; M^2C,\;\ldots
}
by applying the sparse projection $u^\trans$ to each block.  Thus the
solve routines used here require only products by $M$, not by
$M^\trans$.

We can also adapt the EGGSV framework for testing nonsingularity.

\begin{lemma}
  \label{lem:nonsingularity-test}
  Let $N\in \F^{r\times r}$, let $1\leq h\le r$, and suppose a vector
  product by $N$ costs $\softO(L_N)$ operations.  There is a verified
  sparse-block solve test which requires \break
  $\softO\left(rL_N+r^2+r h^{\omega-1}\right)$ operations in $\F$ per
  trial and satisfies the following:
  \begin{enumerate}[label=(\roman*),leftmargin=2.5em]
  \item If $N$ is nonsingular, a trial reports ``nonsingular'' with
    probability at least $1-\eta$, where $\eta$ is the chosen EGGSV
    conditioning failure probability.  Repeating independent trials
    and accepting if any trial succeeds makes this false-negative
    probability arbitrarily small.
  \item If $N$ is singular with nullity $\delta\ge1$, a trial reports
    ``nonsingular'' with probability at most
    $|\F|^{-\delta h}\le |\F|^{-h}$.  If $R$ independent trials are
    used and the algorithm accepts when any trial succeeds, the
    singular false-positive probability is at most $R |\F|^{-h}$.  The
    extension field and the number of repetitions must be chosen so
    that $R|\F|^{-h}$ is sufficiently small.
  \end{enumerate}
\end{lemma}

\begin{proof}
  Choose a random block $C\in\F^{r\times h}$ and run one verified
  sparse-block solve attempt for $NX=C$ using
  Theorem~\ref{thm:param-solve} with block size $h$.  Report
  ``nonsingular'' exactly when the solve attempt returns a matrix $X$
  and the verification $NX=C$ succeeds.

  We use the same bounded EGGSV solve attempt as an algorithmic
  procedure on arbitrary input.  On nonsingular $N$,
  Theorem~\ref{thm:param-solve} provides its Las Vegas success
  behaviour.  On singular $N$, regardless of its internal behaviour,
  the verification $NX=C$ can succeed only if every column of $C$ lies
  in $\im(N)$, and the false-positive probability is at most
  $|\F|^{-\delta h}$.
  
  The cost is the cost of Theorem~\ref{thm:param-solve} with dimension
  $r$, block size $h$, and vector-product cost $L_N$, namely
  $\softO\left(rL_N+r^2+r h^{\omega-1}\right)$. \qed
\end{proof} 

It is straightforward to employ a Wiedemann-style rank routine, using
binary search, singularity testing, and fast solving to determine
rank.  Using \pcite{KalSau91} method may be faster, but some
infrastructure around using the block minimal polynomial would have to
be developed.

\begin{theorem}
  \label{thm:rank-binary-search}
  Let $M\in \F^{m\times m}$ be arbitrary, and suppose one vector
  product $v\mapsto Mv$ costs $\softO(L)$ operations.  There is a
  Monte Carlo algorithm that returns $\rank(M)$ with failure
  probability at most $\epsilon$ and cost
  \tightdisplay{
    \softO\left(mL+m^2+m s^{\omega-1}\right).
  }
  The algorithm uses Wiedemann's rank-by-singularity-reduction
  approach, the PRECONDSXS rank-profile preconditioner of
  \cite{CheEbe02}, and the nonsingularity test of
  Lemma~\ref{lem:nonsingularity-test} from EGGSV.  If the original
  field is too small, the computation is performed over an extension
  field of degree $O(\log m+\log(1/\epsilon))$.
\end{theorem}

\begin{proof}
  We recall \pcite{Wie86} rank-by-search reduction.  For a candidate
  integer $r$, the predicate is $P(r): r\le \rank(M)$. The PRECONDSXS
  conditioning of \cite{CheEbe02}, together with their reductions and
  fast conditioners, supplies fast invertible $P_r,Q_r$ such that the
  leading $r\times r$ block $N_r=(P_rMQ_r)_{1:r,1:r}$ is nonsingular
  with high probability when $r\le\rank(M)$.  If $r>\rank(M)$, every
  $r\times r$ submatrix is singular.

  Test $N_r$ with Lemma~\ref{lem:nonsingularity-test}, using
  $h=\min(r,s)$.  A vector product by $N_r$ is obtained by embedding
  into $m$ coordinates, applying the fast conditioners and $M$, and
  restricting; this costs $\softO(L)$, since a product by $M$ already
  has output size $m$.  Thus each predicate test costs at most
  $\softO\left(mL+m^2+m s^{\omega-1}\right)$, including the case
  $r<s$.  Binary search uses $O(\log m)$ predicate calls.  Repetition
  reduces false negatives when the predicate is true; when it is
  false, the only accepting events are the random-right-hand-side
  false positives of Lemma~\ref{lem:nonsingularity-test}, bounded by
  $R|\F|^{-h}$ for $R$ trials.  Choosing the extension field and
  repetition count so that each predicate has error at most
  $\epsilon/O(\log m)$ and applying a union bound gives total failure
  probability at most $\epsilon$. \qed
\end{proof}

We will employ the adapted methods of \cite{EbeGie06,EbeGie07},
\cite{Wie86}, \cite{CheEbe02}, which we summarize as follows.

\begin{corollary}
  \label{cor:field-engine}
  Let $M\in \F^{m\times m}$ be a finite-field matrix for which one
  vector product costs $\softO(L)$ operations.  The following
  operations are available with the displayed costs:
  \begin{enumerate}[label=(\alph*),leftmargin=2.5em]
  \item If $M$ is nonsingular, compute $M^{-1}C$ for
    $C\in \F^{m\times s}$ with expected
    \tightdisplay{
      \softO\left(mL+m^2+m s^{\omega-1}\right)
      \hspace*{10pt}\mbox{operations in $\F$ (Las Vegas).}
    }
  \item Compute $\rank(M)$ in an expected
    \tightdisplay{
      \softO\left(mL+m^2+m s^{\omega-1}\right)
      \hspace*{10pt}\mbox{operations in $\F$ (Monte Carlo).}
    }
    by binary search on nonsingularity tests, with failure probability
    bounded by the chosen parameter.
  \end{enumerate}
\end{corollary}

Later we will apply Corollary~\ref{cor:field-engine} to specialize to
leading submatrices of the preconditioned sparse matrices.  In
particular, suppose $B_0=B\bmod p\in \F^{a\times a}$, a leading block
of the preconditioned matrix $A'\bmod p$.  A vector product by $B_0$
is implemented by embedding the vector into the first $a$ coordinates
of an $n$-vector, applying $A'\bmod p$, and restricting to the first
$a$ coordinates.  Hence a vector product by $B_0$ costs $\softO(n)$,
and applying $B_0$ to an $a\times s$ block costs $\softO(ns)$ by
applying this vector operation to the columns.  If $B_0$ is
nonsingular, computing $B_0^{-1}C$ for $C\in\F^{a\times s}$ via
Theorem~\ref{thm:param-solve} requires
$\softO\left(an+a^2+a s^{\omega-1}\right)$ operations in $\F$.

The same rank-by-binary-search proof is later applied to the divided
Schur complement operator $T$.  We do not introduce a separate black
box for $T$: Section~\ref{sec:T-oracle} constructs an explicit
$s$-column oracle for $T$ from the original $A$-oracle and from the
verified sparse-block solves with $B_0$.  The rank proof uses $O(k/s)$
calls to this $s$-column black box, together with
$\softO(k^2+k s^{\omega-1})$ additional field operations for sparse
projections and block-Hankel algebra.  This is the counting used in
the cost analysis.

We note that the sparse-block projection, PRECONDSXS, and random
right-hand-side tests are probabilistic and require a sufficiently
large field for precise failure bounds.  If $p$ is too small, choose
$\F=\FF_{p^d}$, for $d=O(\log n+\log(1/\epsilon))$, so that $|\F|$ is
larger than the polynomial-size sets required by the random
conditioners and so that the singular false-positive bound
$R|\F|^{-h}$ is negligible for the planned number $R$ of repetitions.
The algorithm is then run over $\Rhat=\GR(p^2,d)$.  By
Lemma~\ref{lem:extend}, this does not change the Smith multiplicities.
The extension degree contributes only polylogarithmic factors in the
usual algebraic-operation model when $p$ is fixed or word-sized.

\section{Preconditioners for fast SNF computation}

Preconditioners have two distinct roles in our algorithms, which we
outline below.

\subsection{Global Smith/rank-profile preconditioning}

We need fast $P,Q\in\GL_n(\Rhat)$ such that $A'=PAQ$ has a leading
$a\times a$ block $B$ with $B\bmod p$ nonsingular, where
$a=\rank(A\bmod p)$.  The required properties are:

\begin{enumerate}[label=(G\arabic*),leftmargin=3em]
\item $P,Q$ are invertible over $\Rhat$.  This ensures that Smith
  multiplicities are preserved.
\item Their reductions $P_0,Q_0\in\GL_n(\F)$ give generic rank profile
  to $P_0(A\bmod p)Q_0$, or at least make its leading $a\times a$
  submatrix nonsingular.
\item $P,Q$ apply to $n\times s$ blocks in $\softO(ns)$ operations.
\end{enumerate}

The following lemma summarizes our global preconditioning requirement.

\begin{lemma}[Fast generic rank-profile preconditioning]
  \label{lem:global-preconditioner}
  Let $M\in \F^{n\times n}$ have rank $a$.  Over a sufficiently large
  finite field, there are randomized fast invertible conditioners
  $P_0,Q_0\in\GL_n(\F)$ such that $P_0MQ_0$ has a nonsingular leading
  $a\times a$ block with probability at least $1-\epsilon$.  The
  conditioners apply to $n\times s$ blocks in $\softO(ns)$ operations.
  If the field is too small, the same statement holds after extension
  degree $O(\log n+\log(1/\epsilon))$.
\end{lemma}

\begin{proof}
  This is the standard generic-rank-profile/PRECONDSXS conditioning
  theorem used in black-box rank algorithms.  \cite{CheEbe02}
  explicitly formulate PRECONDSXS as the condition that, for a
  requested size $s\le\rank(M)$, the leading $s\times s$ principal
  minor of the preconditioned matrix is nonzero; they also state its
  use in rank computation by binary search
  \cite[Section~3.1]{CheEbe02}.  Their Theorem~3.1 reduces PRECONDSXS
  and generic rank profile to linear-independence preconditioners;
  Theorem~6.2 gives the butterfly routing network, and Theorem~6.3
  gives the associated randomized linear-independence probability
  bound, yielding fast linear-independence conditioners with
  near-linear application cost \cite[Theorems~3.1, 6.2,
  and~6.3]{CheEbe02}.  \pcite{KalSau91} triangular Toeplitz
  conditioners give another standard fast generic-rank-profile family.
  EGGSV use the same kind of fast rank-profile conditioning, together
  with diagonal scalings and butterfly-type conditioners, to make the
  leading block minors and sparse-block Krylov determinants required
  by their efficient projection theorem nonzero over sufficiently
  large fields \cite[Corollary~2.3 and Section~4]{EbeGie07}.  The bad
  choices are zeros of provably nonzero polynomials in the random
  parameters, so a Schwartz--Zippel bound, or the extension-field
  choice used in EGGSV, gives failure probability at most
  $\epsilon$. \qed
\end{proof}

In the Smith algorithm, choose such sparse invertible preconditioners
$P_0,Q_0$ for $M=A\bmod p$ and lift them coefficientwise to
$P,Q\in\GL_n(\Rhat)$.  A lift is invertible over $\Rhat$ because its
determinant is nonzero modulo $p$, hence a unit in the local ring.
Therefore $A'=PAQ$ is Smith-equivalent to $A$, and the leading block
$B\bmod p$ is nonsingular whenever the preconditioning event succeeds.

The event that the leading $a\times a$ block is nonsingular is checked
probabilistically by computing $\rank_\F(B\bmod p)$ using
Corollary~\ref{cor:field-engine}.  If the reported rank is not $a$,
resample $P,Q$.  If a Monte Carlo rank check falsely reports full
rank, this is included in the global Monte Carlo failure probability.
The later solve verifications catch many such mistakes, but the proof
of correctness does not rely on them as deterministic proofs of
nonsingularity.

\subsection{Internal EGGSV sparse-block conditioning}

Inside the finite-field solve algorithms and the solve-based rank
tests we use additional fast conditioners and sparse-block
projections.  These conditioners are not part of the Smith
equivalence, but are internal to the finite-field routines.  These
conditioners must satisfy:

\begin{enumerate}[topsep=0pt,label=(W\arabic*),leftmargin=3em]
\item the relevant block Krylov matrices have full rank or the
  required generic degree profile with high probability;
\item applying the conditioners and sparse projections to $m\times t$
  blocks costs $\softO(mt)$ per Krylov step;
\item the total projection overhead over a sequence of length
  $\softO(m/t)$ is $\softO(m^2)$.
\end{enumerate}

The last item is the key difference from using dense block
projections inside the solve/rank tests.

\section{Carry-safe application of the divided Schur complement}
\label{sec:T-oracle}

This section presents the main new algorithmic ideas.  All operations
are done over $\Rhat$ until an element is known to be divisible by
$p$.

Let
\[
  A'=
  \begin{bmatrix}
    B&C\\
    D&E
  \end{bmatrix},
  \qquad B\in \Rhat^{a\times a},
  \qquad B_0=B\bmod p\in \F^{a\times a}.
\]
Assume $B_0$ is nonsingular.  Let
\tightdisplay{
  J=E-DB^{-1}C=pT,
  \qquad T\in \F^{k\times k},
  \qquad k=n-a.
}

\subsection{Hensel solve modulo $p^2$}

Given $U\in \Rhat^{a\times t}$, solve
\tightdisplay{
  BY=U\bmod p^2
}
as follows.

\begin{enumerate}[label=H\arabic*.,itemsep=2pt,leftmargin=3em]
\item Reduce modulo $p$:
\tightdisplay{
  U_0=U\bmod p,
  \qquad B_0=B\bmod p.
}

\item Solve over $\F$:
  \[
  B_0Y_0=U_0.
\]

\item Lift $Y_0$ coefficientwise to $\widetilde Y_0\in \Rhat^{a\times t}$.

\item Compute the residual in $\Rhat$:
 \tightdisplay{
  W=U-B\widetilde Y_0\bmod p^2.
 }
Because $B_0Y_0=U_0$, the matrix $W$ lies in $p\Rhat^{a\times t}$.

\item Divide by $p$ only now:
  \tightdisplay{
  W_1=W/p\bmod p\in \F^{a\times t}.
  }

\item Solve over $\F$:
  \tightdisplay{
  B_0Y_1=W_1.
  }

\item Lift $Y_1$ to $\widetilde Y_1\in \Rhat^{a\times t}$ and return
  \tightdisplay{
    Y=\widetilde Y_0+p\widetilde Y_1.
  }
\end{enumerate}

\begin{lemma}[Correctness of the Hensel solve]
  \label{lem:hensel}
  The matrix $Y$ returned above satisfies $BY=U\bmod p^2$.
\end{lemma}

\begin{proof}
  By construction,
  \[
    U-B\widetilde Y_0=pW_1\bmod p^2.
  \]
  Also $B_0Y_1=W_1$ over $\F$, so
  \[
    pB\widetilde Y_1=pB_0Y_1=pW_1\bmod p^2.
  \]
  Therefore
  \[
    B(\widetilde Y_0+p\widetilde Y_1)
    =B\widetilde Y_0+pB\widetilde Y_1
    =B\widetilde Y_0+pW_1
    =U\bmod p^2. \tag*{\qed}
  \]
\end{proof}

\subsection{The carry trap}

Suppose we write digit expansions
\tightdisplay{
  B=B_0+pB_1,
  \qquad U=U_0+pU_1.
}
It is tempting, but incorrect, to set
\tightdisplay{
  W_1=U_1-B_1Y_0.
}
The missing term is the carry from the product $B_0Y_0$ in
representatives.  If
\tightdisplay{
  B_0Y_0=U_0+pK
}
in canonical representatives, then the correct residual is
\tightdisplay{
  W_1=U_1-B_1Y_0-K\bmod p.
}
Computing
\tightdisplay{
  W_1=(U-B\widetilde Y_0)/p\bmod p
}
after the multiplication in $\Rhat$ automatically includes the
carry~$K$.

\subsection{Oracle for $X\mapsto TX$}

Let $X\in \F^{k\times t}$.  The block oracle for $TX$ is:

\begin{enumerate}[label=T\arabic*.,leftmargin=3em]
\item Lift $X$ to $\widetilde X\in \Rhat^{k\times t}$.

\item Compute over $\Rhat$:
\[
  A'\begin{bmatrix}0\\\widetilde X\end{bmatrix} =
  \begin{bmatrix}U\\V\end{bmatrix}
  =
  \begin{bmatrix}C\widetilde X\\E\widetilde X\end{bmatrix}
  \bmod p^2.
\]

\item Use the Hensel solve above to compute $Y$ satisfying
\[
  BY=U\bmod p^2.
\]
This step involves two solves with $B_0$ over $\F$ and one carry-safe
residual computation in $\Rhat$.

\item Compute over $\Rhat$:
\[
  Z=A'
  \begin{bmatrix}
    -Y\\\widetilde X
  \end{bmatrix} =
  \begin{bmatrix}
    0\\J\widetilde X
  \end{bmatrix}
  \bmod p^2.
\]

\item The lower block of $Z$ is divisible by $p$.  Return
\[
  TX=(\text{lower block of }Z)/p\bmod p.
\]
\end{enumerate}

\begin{lemma}[Correctness of the $T$-oracle]\label{lem:Toracle}
  The oracle returns $TX$, where $T=J/p\bmod p$.
\end{lemma}

\begin{proof}
  Step T2 gives $U=C\widetilde X$.  Step T3 computes
  $Y=B^{-1}U\bmod p^2$.  Therefore
  \[
    A'\begin{bmatrix}
      -Y\\
      \widetilde X
    \end{bmatrix} =
    \begin{bmatrix}
      -BY+C\widetilde X\\
      -DY+E\widetilde X
    \end{bmatrix}
    =
    \begin{bmatrix}
      0\\
      (E-DB^{-1}C)\widetilde X
    \end{bmatrix}
    =
    \begin{bmatrix}
      0\\
      J\widetilde X
    \end{bmatrix} \bmod p^2.
  \]
  Since $J=pT$, the lower block is divisible by $p$, and division by
  $p$ gives $TX$ over $\F$. The result is independent of the chosen
  lift of $X$: replacing $\widetilde X$ by $\widetilde X+pW$ changes
  $J\widetilde X$ by $JpW=p^2TW=0$ modulo $p^2$.  \qed
\end{proof}

\section{The main algorithm for SNF modulo $p^2$}

We now present our algorithm for sparse Smith Normal Form modulo
$p^2$.

\medskip\noindent
Input: $A\in \R^{n\times n}$ with fast sparse matrix-vector products.\\
Output: $(r_0,r_1,r_2)$ such that $\SNF(A)=\diag(1^{r_0},p^{r_1},0^{r_2})$.

\begin{enumerate}[label=A\arabic*.,itemsep=3pt,leftmargin=3em]
\item Choose an extension $\Rhat=\GR(p^2,d)$ if needed so that the
  residue field $\F$ is large enough for the sparse-block
  randomization.

\item Compute $a=\rank_\F(A\bmod p)$ and set $r_0=a$ and $k=n-a$.
  
\item If $k=0$, return $(n,0,0)$.
  
\item If $a=0$, then $A\equiv 0\bmod p$, so $A=pT$ for
  $T\in \F^{n\times n}$.  Compute
  \tightdisplay{
    r_1=\rank_\F(A/p\bmod p),
  }
  and return $(0,r_1,n-r_1)$.
  
\item Choose fast unimodular $P,Q\in\GL_n(\Rhat)$ such that, for
  \[
    A'=PAQ=\begin{bmatrix}B&C\\D&E\end{bmatrix},
    \qquad B\in \Rhat^{a\times a},
  \]
  the matrix $B_0=B\bmod p$ is nonsingular.  Verify this condition
  probabilistically by computing $\rank_\F(B_0)$.  If the check
  reports failure, resample.

\item Choose the working block size
\tightdisplay{
  t=\min\left\{a,k,\ceil{n^{1/(\omega-1)}}\right\}.
}
This is the number of simultaneous right hand sides used in all
subsequent black-box queries to the Schur-complement oracle.  Thus the
outer rank computation will query the black box on matrices
$X\in\F^{k\times t}$, and each such query reduces internally to solves
with coefficient matrix $B_0$ and $t$ right hand sides.  The bounds
$t\le k$ and $t\le a$ ensure that the same block size is valid both
for the $k\times k$ matrix $T$ and for the internal $a\times a$ solves
with $B_0$.

\smallskip
The cap $t\le \ceil{n^{1/(\omega-1)}}$ is the value used in the
complexity analysis below; it balances the two main block-size
dependent terms, namely the repeated sparse/block-access term
$akn/t$ and the dense block-arithmetic term $ak t^{\omega-2}$.

\item Use the carry-safe oracle of Section~\ref{sec:T-oracle} to
supply, for every query matrix $X\in\F^{k\times t}$, the block product
\tightdisplay{
  X\longmapsto TX,\qquad
  T=(E-DB^{-1}C)/p\bmod p.
}
The matrix $T$ is not formed explicitly.

\item Compute $b=\rank_\F(T)$ using the rank-by-binary-search part of
  Corollary~\ref{cor:field-engine}, with the above block size $t$ and
  with access to $T$ provided by the oracle in Step~A7.
  
\item Return $(r_0,r_1,r_2)=(a,b,k-b)$.
\end{enumerate}

Correctness is established in the following theorem.

\begin{theorem}[Correctness]
  \label{thm:correctness}
  Assume the rank computations and sparse-block solves return correct
  results and the preconditioner in Step A5 succeeds.  Then the
  algorithm returns the Smith multiplicities of $A$ modulo $p^2$.
\end{theorem}

\begin{proof}
  Step A2 computes $a=\rank(A\bmod p)$, which equals the number of
  unit invariant factors.  Hence $r_0=a$.

  The preconditioners $P,Q$ are invertible over $\Rhat$, so $A$ and
  $A'=PAQ$ have the same Smith multiplicities over $\Rhat$, and
  therefore the same multiplicities as over $\R$.  Step A5 ensures
  that the leading block $B$ is invertible over $\Rhat$.  By
  Lemma~\ref{lem:schur},
  \tightdisplay{
    A'\sim \begin{bmatrix}B&0\\0&J\end{bmatrix},
    \qquad J=E-DB^{-1}C,
  }
  and $J=pT$ for a matrix $T\in \F^{k\times k}$.  The block $B$
  contributes exactly $a$ unit invariant factors.  By
  Lemma~\ref{lem:divided-schur-rank}, the block $J=pT$ contributes
  $\rank_\F(T)$ invariant factors equal to $p$ and $k-\rank_\F(T)$
  invariant factors equal to $0$.  Step A8 computes this rank, so the
  returned multiplicities are correct. \qed
\end{proof}

The complexity is established as follows.  Let $\omega$ be a feasible
exponent for dense matrix multiplication.  Let $\theta=1/(\omega-1)$
and $\tau=3-\theta=3-1/(\omega-1)$.  Note that for $2<\omega<3$, we have
$2<\tau<\omega$.

\paragraph*{Cost of one $T$-block application}

A block application $X\mapsto TX$ with $X\in \F^{k\times t}$ uses:

\begin{itemize}[leftmargin=2.1em]
\item a constant number of block products by $A'$ over $\Rhat$ with $t$
  columns, costing $\softO(nt)$; this includes the products used to
  form $C\widetilde X$, to compute the Hensel residual
  $U-B\widetilde Y_0$, and to form the final vector
  $A'[-Y;\widetilde X]$;
\item two solves with $B_0\in \F^{a\times a}$ and $t$ right hand sides;
\item sparse-block projection and matrix-polynomial work inside these
  solves.
\end{itemize}

A product by $B_0$ with $t$ columns costs $\softO(nt)$: embed the
$a\times t$ block into the first $a$ coordinates of an $n\times t$
block, apply $A'\bmod p$, and restrict to the first $a$ coordinates.
Therefore, by Corollary~\ref{cor:field-engine}, one solve costs
\[
  \softO\left(\frac{a}{t}\,nt+a^2+a t^{\omega-1}\right)
  =
  \softO\left(an+a^2+a t^{\omega-1}\right).
\]
Thus one $T$-block application costs
\[
  \softO\left(nt+an+a^2+a t^{\omega-1}\right).
\]
Since $a\le n$, the $a^2$ term is bounded by $an$, but we keep it visible for the derivation.

\subsection{Cost of computing rank $T$}

The rank computation for $T\in \F^{k\times k}$ is the solve-based
binary-search rank computation of
Theorem~\ref{thm:rank-binary-search}.  Each nonsingularity test is
applied to a leading submatrix of a fast preconditioning of $T$.  A
block product by such a tested submatrix is obtained by embedding into
the $k$ coordinates, applying the fast preconditioners, applying the
$T$-oracle, and restricting.  The rank-by-binary-search proof uses
$O(k/t)$ such $t$-column oracle calls, plus
$\softO(k^2+k t^{\omega-1})$ additional finite-field work.  Using the
$T$-block cost just derived, the cost after the Schur preconditioning
is
\begin{align*}
  C(a,k,t)
  & = \softO\left(
    \frac{k}{t}\left(nt+an+a^2+a t^{\omega-1}\right)
    +k^2+k t^{\omega-1}
    \right) \\
  & = \softO\left(
    kn+\frac{akn}{t}+\frac{a^2k}{t}
    +ak t^{\omega-2}+k^2+k t^{\omega-1}
    \right).
\end{align*}
Because $a\le n$,
\[
  \frac{a^2k}{t}\le \frac{akn}{t}.
\]
Thus
\[
 C(a,k,t)
 = \softO\left(
     kn+\frac{akn}{t}+ak t^{\omega-2}+k^2+k t^{\omega-1}
   \right).
\]
The initial rank computation, the verification of $B_0$, and the
setup/profile-check work for the fast preconditioners add
$\softO(n^2)$ operations, which is below the final bound.
Applications of the internal conditioners during finite-field rank and
solve routines are already included in the vector-product and $m^2$
sparse-projection accounting of Section~\ref{sec:finite-field-engine}.

\subsection{Choice of block size}

Choose
\[
  t=\min\left\{a,k,\ceil{n^\theta}\right\}, \qquad
  \theta=1/(\omega-1).
\]
We prove that the cost is $\softO(n^\tau)$ uniformly over all ranks
$a$.

\begin{lemma}[Uniform bound]\label{lem:uniform-bound}
  For all $a,k\ge0$ with $a+k=n$, and with $t$ chosen as above when
  $a,k>0$,
  \[
    C(a,k,t)=\softO(n^\tau),
    \qquad \tau=3-1/(\omega-1).
  \]
\end{lemma}

\begin{proof}
  The terms $kn$, $k^2$, and the initial $n^2$ overhead are all
  $O(n^2)$.  Also
  \[
    k t^{\omega-1}\le \softO\left(n\cdot (n^\theta)^{\omega-1}\right)=\softO(n^2).
  \]
  It remains to bound
  \[
    \frac{akn}{t}
    \quad\text{and}\quad
    ak t^{\omega-2}.
  \]
  
  First suppose $\min\{a,k\}\ge n^\theta$.  Then $t=\Theta(n^\theta)$,
  and since $ak\le n^2$,
  \[
    \frac{akn}{t}\le \softO(n^{3-\theta})=\softO(n^\tau).
  \]
  Similarly,
  \[
    ak t^{\omega-2}
    \le \softO(n^2\cdot n^{\theta(\omega-2)})
    =\softO\left(n^{2+(\omega-2)/(\omega-1)}\right)
    =\softO\left(n^{3-1/(\omega-1)}\right)
    =\softO(n^\tau).
  \]

  Now suppose $m=\min\{a,k\}<n^\theta$.  Then $t=m$.  Since
  $ak/m=\max\{a,k\}\le n$,
  \[
    \frac{akn}{t}=\frac{akn}{m}\le n^2.
  \]
  Also
  \[
    ak t^{\omega-2}=ak m^{\omega-2}
    =\max\{a,k\}\, m^{\omega-1}
    \le \softO\left(n\,(n^\theta)^{\omega-1}\right)
    =\softO(n^2).
  \]
  Since $\tau>2$ for $\omega>2$, these terms are also
  $O(n^\tau)$. \qed
\end{proof}

The overall cost is summarized as follows.

\begin{theorem}[Cost of computing the Smith multiplicities]
  \label{thm:main-complexity}
  Using the sparse-block algorithm summarized in
  Section~\ref{sec:finite-field-engine} above and the fast
  rank-profile preconditioner of
  Lemma~\ref{lem:global-preconditioner}, suppose the product
  $v\mapsto Av$ costs $\softO(n)$ operations.  An $n\times s$ product
  is obtained by applying this oracle to each column, and costs
  $\softO(ns)$.

  Then the Smith multiplicities of any $A\in\R^{n\times n}$ can be
  computed by a randomized algorithm which returns the correct
  multiplicities with probability at least $1-\epsilon$ and whose
  expected number of operations over the working ring $\Rhat$ is
  \[
    \softO\left(n^{3-1/(\omega-1)}\right).
  \]
  For fixed or word-sized $p$, and with
  $[\Rhat:\R]=O(\log n+\log(1/\epsilon))$, this is also
  \[
    \softO\left(n^{3-1/(\omega-1)}\right)
  \]
  operations over the base ring $\R$.
\end{theorem}

\begin{proof}
  Correctness follows from Theorem~\ref{thm:correctness}.  The cost of
  the main rank computation is bounded by
  Lemma~\ref{lem:uniform-bound}.  The initial rank computations, the
  $B_0$ check, and the setup/profile-check work for the global
  preconditioners cost $\softO(n^2)$.  Applications of global and
  internal conditioners inside the finite-field rank and solve
  routines are part of the vector-product and sparse-projections in
  Section~\ref{sec:finite-field-engine}.  The field extension and
  failure-probability amplification add only polylogarithmic factors
  hidden in the soft-O. \qed
\end{proof}

\section{Generalization to SNF modulo $p^e$}

The Schur/divide identity used above can be extended to higher
powers of $p$.  Let $A$ be a matrix over $\ZZ/p^e\ZZ$, or over the
corresponding unramified Galois ring of length $e$.  After row and
column preconditioning, suppose
\[
  A=
  \begin{bmatrix}
    B & C\\
    D & E
  \end{bmatrix},
  \qquad
  B\bmod p \text{ nonsingular},
\]
where the size of $B$ is $\rank(A\bmod p)$.  Then $B$ is invertible
and $A$ is equivalent to
\[
  \begin{bmatrix}
    B & 0\\
    0 & J
  \end{bmatrix},
  \qquad
  J=E-DB^{-1}C .
\]
Moreover $J\equiv 0\bmod p$, so we can divide $J$ by $p$ and continue
over the shorter ring $\ZZ/p^{e-1}\ZZ$.  Thus the same algebra gives a
recursive description of the Smith multiplicities.  This recursion
does not carry over to the complexity result of this paper.  For $e=2$
there is only one divided Schur complement to rank, and its block
products are obtained through a single layer of carry-safe solves with
the unit block $B$.  For $e\ge3$, the next divided Schur complement
already has products that are themselves defined through lower-level
solves.  A direct black-box recursion therefore nests the rank and
solve procedures, rather than preserving the single-level amortization
used above.

Consequently, the straightforward recursive use of the present method
does not give a worst-case bound below the dense $\softO(n^\omega)$
benchmark for higher powers of $p$.  The results of this paper are
therefore stated only for length-two rings.   The corresponding
sub-$n^\omega$ sparse SNF problem for $\ZZ/p^e\ZZ$, $e\geq 3$,
remains open.

\section{Conclusion}

We have combined a number of modern iterative block matrix tools with
rank-profile Schur splitting over $\R=\ZZ/p^2\ZZ$ and carry-safe
Hensel solves to create a new Monte Carlo algorithm for sparse Smith
form modulo $p^2$ which requires an expected
\[
  \softO(n^{3-\frac{1}{\omega-1}})
  \hspace*{10pt}
  \mbox{operations in $\R$.}
\]
When using standard cubic matrix arithmetic with $\omega=3$ this cost
is $\softO(n^{2.5})$ operations in $\R$, whereas with the best
currently known $\omega\approx 2.371339$ our algorithm requires
$\softO(n^{2.270786})$ operations in $\R$, which has exponent well
below $\omega$.

The key points are that the Schur complement is computed only through
a block oracle, all divisions by $p$ occur only after divisibility is
established in the ring, and sparse-block projections are used inside
the EGGSV solve tests so that the projection cost is $\softO(m^2)$
rather than $\softO(m^2t)$.

For matrices modulo $p^e$ the same Schur/divide approach is valid, but
the simple recursive sparse-block complexity analysis hits an
$n^\omega$ barrier already at $e=3$.  Getting past this appears to
require a new multilevel blocking or filtered-module Krylov idea, not
merely a black-box recursive invocation of the $p^2$ algorithm.

The algorithm presented here is of the Monte Carlo type, with a
controllably small probability of error. It would be interesting to
find a Las Vegas algorithm with this same complexity, perhaps
following the certified rank techniques of \cite{EbeGie07}.

\section*{Acknowledgements}
The author would like to thank the anonymous referees for their input.  The author acknowledges support of the Natural Sciences and Engineering Research Council (NSERC), Canada.

\patchcmd{\thebibliography}
  {\list}
  {\vspace{-3em}\list}
  {}
  {}
  
\begingroup
\renewcommand{\bibfont}{\small}
\setlength{\bibsep}{0pt}
\bibliographystyle{plainnat}

\renewcommand{\doi}[1]{%
  \href{https://doi.org/#1}{doi:\footnotesize\nolinkurl{#1}}%
}

\bibliography{snfp2}

@inproceedings{ElsGie12,
  author    = {Mustafa Elsheikh and Mark Giesbrecht and
               Andy Novocin and B. David Saunders},
  title     = {Fast Computation of {S}mith Forms of Sparse Matrices Over Local Rings},
  booktitle = {Proceedings of the International Symposium on
               Symbolic and Algebraic Computation (ISSAC 2012)},
  pages      = {146--153},
  year       = {2012},
  doi        = {10.1145/2442829.2442853}
}

@inproceedings{EbeGie06,
  author    = {Eberly, Wayne and Giesbrecht, Mark and Giorgi, Pascal and Storjohann, Arne and Villard, Gilles},
  title     = {Solving Sparse Rational Linear Systems},
  booktitle = {Proc. 2006 International Symposium on Symbolic and Algebraic Computation},
  series    = {ISSAC '06},
  pages     = {63--70},
  year      = {2006},
  doi       = {10.1145/1145768.1145785}
}

@inproceedings{EbeGie07,
  author    = {Eberly, Wayne and Giesbrecht, Mark and Giorgi, Pascal and Storjohann, Arne and Villard, Gilles},
  title     = {Faster Inversion and Other Black Box Matrix Computations Using Efficient Block Projections},
  booktitle = {Proceedings of the 2007 International Symposium on Symbolic and Algebraic Computation},
  series    = {ISSAC '07},
  pages     = {143--150},
  publisher = {ACM},
  address   = {New York, NY, USA},
  year      = {2007},
  doi       = {10.1145/1277548.1277569}
}

@article{CheEbe02,
  author  = {Chen, L. and Eberly, Wayne and Kaltofen, Erich and Saunders, B. David and Turner, William J. and Villard, Gilles},
  title   = {Efficient Matrix Preconditioners for Black Box Linear Algebra},
  journal = {Linear Algebra and its Applications},
  volume  = {343--344},
  pages   = {119--146},
  year    = {2002},
  doi     = {10.1016/S0024-3795(01)00472-4}
}

@article{Wie86,
  author  = {Wiedemann, Douglas H.},
  title   = {Solving Sparse Linear Equations over Finite Fields},
  journal = {IEEE Transactions on Information Theory},
  volume  = {32},
  number  = {1},
  pages   = {54--62},
  year    = {1986},
  doi     = {10.1109/TIT.1986.1057137}
}

@inproceedings{KalSau91,
  author    = {Kaltofen, Erich and Saunders, B. David},
  title     = {On {Wiedemann}'s Method of Solving Sparse Linear Systems},
  booktitle = {Applied Algebra, Algebraic Algorithms and Error-Correcting Codes},
  series    = {Lecture Notes in Computer Science},
  volume    = {539},
  pages     = {29--38},
  publisher = {Springer},
  year      = {1991},
  doi       = {10.1007/3-540-54522-0_93}
}

@article{ChoCol82,
    title = {Algorithms for the Solution of Systems of Linear Diophantine Equations},
    author = {Tsu-Wu J. Chou and G. E. Collins},
    publisher = {SIAM},
    year = {1982},
    journal = {SIAM Journal on Computing},
    volume = {11},
    number = {4},
    pages = {687--708},
    doi = "10.1137/0211057"
}

@book{Hu69,
  title={Integer programming and network flows},
  author={Hu, T. C.},
  year={1969},
  publisher={Addison-Wesley},
  address = {Reading, Mass.},
}

@incollection{DumHec03,
    author = {J.-G. Dumas and F. Heckenbach and B. D. Saunders and V. Welker},
    title = {Computing Simplicial Homology Based on Efficient {Smith} Normal Form Algorithms},
    booktitle = {Algebra, Geometry, and Software Systems},
    publisher = {Springer},
    pages = {177--206},
    year = {2003},
    doi = "10.1007/978-3-662-05148-1_10"
}

@book{New72,
 author = {Newman, M.},
 title = {Integral Matrices},
 year = {1972},
 publisher = {Academic Press},
 address = {New York, NY, USA},
}

@article{Haf89,
     title = {A Rigorous Subexponential Algorithm For Computation of Class Groups},
     author = {Hafner, J. L. and McCurley, K. S.},
     journal = {J. American Mathematical Society},
     volume = {2},
     number = {4},
     pages = {837--850},
     year = {1989},
     publisher = {American Mathematical Society},
     doi = {10.1090/S0894-0347-1989-1002631-0},
}

@book{Kai80,
  title={Linear Systems},
  author={Kailath, T.},
  year={1980},
  publisher={Prentice-Hall},
  address = {Englewood Cliffs, NJ}
}

@article{Smi61,
     title = {On Systems of Linear Indeterminate Equations and Congruences},
     author = {Smith, H. J. S.},
     journal = {Phil. Trans. of the Royal Society of London},
     volume = {151},
     number = {},
     pages = {293--326},
     year = {1861},
     publisher = {The Royal Society},
     doi      = {10.1098/rstl.1861.0016},
}

@article{Kap49,
     title = {Elementary Divisors and Modules},
     author = {Kaplansky, I.},
     journal = {Transactions of the American Mathematical Society},
     volume = {66},
     number = {2},
     pages = {464--491},
     year = {1949},
     publisher = {American Mathematical Society},
     doi      = {10.2307/1990591},
}

@article{Sta16,
 title = {Smith normal form in combinatorics},
 volume = {144},
 doi = {10.1016/j.jcta.2016.06.013},
 journal = {Journal of Combinatorial Theory, Series A},
 author = {Stanley, Richard P.},
 year = {2016},
 pages = {476--495},
}

@inproceedings{AlmDua25,
author = {Josh Alman and Ran Duan and Virginia Vassilevska Williams and Yinzhan Xu and Zixuan Xu and Renfei Zhou},
title = {More Asymmetry Yields Faster Matrix Multiplication},
booktitle = {Proceedings of the 2025 Annual ACM-SIAM Symposium on Discrete Algorithms (SODA)},
pages = {2005--2039},
doi = {10.1137/1.9781611978322.63},
year = 2025,
}

@article{GiuGhr16,
 title = {Two’s company, three (or more) is a simplex},
 author = {Giusti, Chad and Ghrist, Robert and Bassett, Danielle S.},
 volume = {41},
 doi = {10.1007/s10827-016-0608-6},
 number = {1},
 journal = {Journal of Computational Neuroscience},
 year = {2016},
 pages = {1--14},
}

@inproceedings{Sto96,
  author    = {Storjohann, Arne},
  title     = {Near Optimal Algorithms for Computing {Smith} Normal Forms of Integer Matrices},
  booktitle = {Proceedings of the 1996 International Symposium on Symbolic and Algebraic Computation},
  series    = {ISSAC '96},
  pages     = {267--274},
  publisher = {ACM},
  year      = {1996},
  doi       = {10.1145/236869.237084}
}
\endgroup

\end{document}